\documentclass[onecolumn,10pt,journal]{IEEEtran}
\usepackage{color,rotating,pdflscape}
\usepackage{cite}
\usepackage{amsbsy,extarrows}
\usepackage{graphicx,multirow,bm}
\usepackage[table]{xcolor}
\usepackage{ulem}
\usepackage{extarrows}

\usepackage{indentfirst}
\usepackage{color}

\usepackage{ulem}
\usepackage{amsfonts,amsmath,amssymb,amsthm}
\usepackage{latexsym,amscd}
\usepackage{amsbsy}
\usepackage{graphicx,multirow,bm}
\usepackage{algorithm}
\usepackage{algorithmicx}
\usepackage{algpseudocode}
\usepackage{xcolor}
\usepackage{tikz}
\usepackage{diagbox}
\usepackage{indentfirst,multicol}
\usepackage{times}
\usepackage{bm}
\usepackage{amsmath}
\usepackage{cases}

\newtheorem{Theorem}{Theorem}

\newtheorem{Definition}{Definition}
\newtheorem{Example}{Example}

\newtheorem{Remark}{Remark}
\newtheorem{Lemma}{Lemma}

\usepackage{pgf}
\usepackage{tikz}
\usetikzlibrary{matrix,chains,positioning,arrows,calc,decorations,plotmarks,patterns,trees,fit,backgrounds,shapes.multipart,topaths}
\usepackage{pgfplots}
\usetikzlibrary{external}                       
\usepackage{scalefnt}
\pgfplotsset{compat=1.3}
\tikzstyle{help lines}=[black!20,dashed]

\pgfplotscreateplotcyclelist{mylist}{red,blue,black,yellow,brown}

\pgfdeclarepatternformonly{my north east lines}{\pgfqpoint{-1pt}{-1pt}}{\pgfqpoint{8pt}{8pt}}{\pgfqpoint{6pt}{6pt}}%
{
	\pgfsetlinewidth{0.4pt}
	\pgfpathmoveto{\pgfqpoint{0pt}{0pt}}
	\pgfpathlineto{\pgfqpoint{6.1pt}{6.1pt}}
	\pgfusepath{stroke}
}
\definecolor{light_gray}{rgb}{0.6,0.6,0.6}
\definecolor{awgray}{rgb}{0.7,0.7,0.7}
\definecolor{awgray_dark}{rgb} {0.4,0.4,0.4}

\tikzset{
	>=stealth',
	mycircle/.style={circle, draw=gray, very thick},
	mycircle_small/.style={circle,draw=awgray_dark,fill = awgray_dark, inner sep=0,minimum size=.6em},
	mycircle_small_black/.style={circle,draw=black,fill = black, inner sep=0,minimum size=.6em},
	mybox/.style={rectangle,rounded corners,draw=black, thick,text width=1em,minimum height=4em,minimum width=4em,text centered},
	mybox_small/.style={rectangle,rounded corners,draw=black, thick,text width=1em,minimum height=2em,minimum width=2em,text centered},
	mybox_vec/.style={rectangle,rounded corners,draw=black, thick,text width=1em,minimum height=0.7em, minimum width=4em,text centered},
	mybox_vec_short/.style={rectangle,rounded corners,draw=black, thick,text width=1em,minimum height=0.7em, minimum width=2em,text centered},
	pil/.style={->, thick, shorten <=2pt, shorten >=2pt,},
}

\usetikzlibrary{arrows,shapes,chains}

\begin{document}

\title{A New Cooperative Repair Scheme with $k+1$ Helper Nodes for $(n,k)$ Hadamard MSR codes  with Small Sub-packetization
\author{Yajuan Liu, Han Cai,~\IEEEmembership{Member,~IEEE}, and Xiaohu Tang, \IEEEmembership{Senior Member, IEEE}
\thanks{
Y. Liu, H. Cai, and X. Tang are with the Information Security and National Computing Grid Laboratory, Southwest Jiaotong University, Chengdu, China (email: yjliu@my.swjtu.edu.cn, hancai@swjtu.edu.cn, xhutang@swjtu.edu.cn).
}}}

\maketitle

\begin{abstract}

Cooperative repair model is an available technology to deal with multiple node failures in distributed storage systems. Recently, explicit constructions of cooperative MSR codes were given by Ye (IEEE Transactions on Information Theory, 2020) with sub-packetization level $(d-k+h)(d-k+1)^n$. Specifically, the sub-packetization level is $(h+1)2^n$ when $d=k+1$. In this paper, we propose a new cooperative repair scheme by means of the inter-instance and intra-instance pairing inherited from the perfect code which reduces the sub-packetization to $2^n$ when $(h+1)|2^n$ and $(2\ell+1)2^n$ when $h+1=(2\ell+1)2^m$ for $m\ge 0$, $\ell\ge 1$ with $d=k+1$ helper nodes.  That is to say,  the sub-packetization is $h + 1 $ times  or  $2^m$ times less than Ye's. It turned out to be the best result so far known.

\end{abstract}

\begin{IEEEkeywords}
Cooperative repair, MDS codes, sub-packetization, perfect code, optimal repair bandwidth.
\end{IEEEkeywords}

\section{Introduction}

With the widespread deployment of large-scale distributed storage systems, such as Facebook's coded Hadoop, Google Colossus, and Microsoft Azure,  reliability is becoming one of the major concerns
so that  the redundancy is imperative. In Particular,  the
maximum distance separable (MDS) code, such as Reed-Solomon code \cite{Polynomial}, offers maximum reliability at the same redundancy level and thus is the most attractive. However,
when a storage node fails, the conventional MDS codes adopt a  naive recovery strategy which first reconstructs the original file and then repairs the failed node.
It  results in a large \textit{repair bandwidth}, which is defined as the amount of data downloaded to repair failed nodes.

In \cite{Network-coding}, regenerating code  introduced by Dimakis \textit{et al.} is shown to achieve the best tradeoff between  the repair bandwidth and storage overhead.  The minimum storage regenerating (MSR) code is one of the two most important regenerating codes, which can maintain the systematic MDS property and has the optimal repair bandwidth.  In the past decades, various MSR code constructions  have been proposed, refer to \cite{balaji2018erasure} for details.

The aforementioned MSR codes  only focus on a single node failure. Whereas, multiple node failures in large-scale  distributed storage systems are the norm rather than the exception. In this scenario, to reduce the management cost,  the repair mechanism is   triggered   only after the total amount of failed nodes reaches a given threshold. So far,  there are two models for repairing multiple node failures. The first one is the centralized repair model, where all the failed nodes are recreated  at a data center. The other is the cooperative repair  model, where the new nodes respectively download data from helper nodes and then communicate with each other to finish the repair process.  It is proved in \cite{Cooperative-repair} that cooperative repair model is stronger than centralized repair model since the optimality of an MDS code under the former implies its optimality under the latter. Besides, the cooperative repair model is more suitable for distributed storage systems owing to the distributed architecture. Therefore, we concentrate on the cooperative repair model in this paper.

The $(n,k)$ MDS code $\mathcal{C}$ consists of $k$ systematic nodes and $r=n-k$ parity nodes, each node storing $N$ symbols. It is a typical high rate storage code in distributed storage system such that the data of any $k$ out of $n$ nodes suffice to reconstruct the whole source data. Specifically,  $N$ is said to be \textit{sub-packetization} level of code $\mathcal{C}$ in the literature. Assume that there are $h$ failed nodes.  To repair them, one needs to connect $d\ge k$ helper nodes.
Under the  cooperative repair model, the repair process  is composed of two phases \cite{Cooperative-regenerating-codes}:
\begin{enumerate}
	\item\label{step2} \textbf{Download phase:} Any failed node downloads $\beta_1$ symbols from each of $d$ helper nodes, respectively;
	
	\item\label{step3} \textbf{Cooperative  phase:} For any two failed nodes,  each transfers $\beta_2$ symbols to the other.
\end{enumerate}
Accordingly,  the repair bandwidth is $\gamma=h(d\beta_1+(h-1)\beta_2)$.

In \cite{Cooperative-recovery} and \cite{Cooperative-regenerating-codes},  the repair bandwidth of MDS array code $\mathcal{C}$ is proved to be lower bounded by
\begin{eqnarray}\label{Eqn_Lower_Bound}
\gamma\ge \frac{h(d+h-1)N}{d-k+h}.
\end{eqnarray}
Especially, the optimal repair bandwidth $ \frac{h(d+h-1)N}{d-k+h}$ in \eqref{Eqn_Lower_Bound} is achieved only for $\beta_1= \beta_2=\frac{N}{d-k+h}$.
In \cite{Cooperative-repair}, Ye and Barg introduced the first explicit construction with optimal cooperative repair property for sub-packetization level $N=((d-k)^{h-1}(d-k+h))^{n \choose h}$. Follow-up works committed to reduce the sub-packetization of MDS codes with the optimal cooperative repair property.  In \cite{Explicit-Constructions-of-Optimal-Access}, Zhang \textit{et al.}  proposed a code with the optimal access property while decreasing the sub-packetization to $(d-k+h)^{n \choose h}$. Remarkably, the recent work of Ye \cite{New-Constructions} dramatically lower the sub-packetization  to $(d-k+h)(d-k+1)^n$.

Actually, the construction in \cite{New-Constructions} is an extension of the Hadamard MSR code  in \cite{Hadamard-MSR} by space sharing its $d-k+h$ instances. Particularly, the key technique thereby is to pair the instances.  Inspired by \cite{New-Constructions}, in this paper we utilize not only the inter-instance pairing, but also the intra-instance pairing  inherited from the Hamming code, i.e., perfect code.
As a consequence, when $d=k+1$, we further decrease  the sub-packetization  to $2^n$ when $(h+1)|2^n$ and $(2\ell+1)2^n$ when $h+1=(2\ell+1)2^m$ for $ \ell\ge 1$.  That is to say,  the sub-packetization is $h + 1 $ times  or  $2^m$ times less than the one obtained in \cite{New-Constructions}.
Table \ref{comparison} illustrates the comparison. It should be noted that our  intra-instance pairing  only works for the case $d=k+1$ due to the nonexistence of  perfect code in the other cases \cite{T-K-Moon}.

\begin{table}[h]
	\begin{center}\caption{comparison with existing constructions of $(n,k)$ MSR codes with optimal cooperative repair property}\label{comparison}
		\renewcommand\arraystretch{1.2}
		\begin{tabular}{|c|c|c|}
			\hline
			& sub-packetization $N$  & helper nodes \\
			\hline
			Ye and Barg \cite{Cooperative-repair}& $((d-k)^{h-1}(d-k+h))^{n \choose h}$ & $k\le d \le n-h$  \\
			\hline
			Zhang \textit{et al.} \cite{Explicit-Constructions-of-Optimal-Access}& $(d-k+h)^{n \choose h}$ & $k\le d \le n-h$\\
             \hline
            Ye \cite{New-Constructions} & $(d-k+h)(d-k+1)^n$& $k\le d \le n-h$ \\
			\hline
			This paper &
			$\left\{\begin{array}{ll}
			     2^n, & \mathrm{if}~(h+1)|2^n\\
			    (2\ell+1)2^n, & \mathrm{if}~ h+1=(2\ell+1)2^m,  \ell\ge 1
			    \end{array}
			\right.$  &  $d=k+1$\\
			\hline
		\end{tabular}
	\end{center}
\end{table}

The remainder of this paper is organized as follows. In Section \ref{Preliminary}, some necessary preliminaries of Hadamard MSR code and Hamming code are reviewed. In Section \ref{basic repair},  a general repair principle is proposed. Then following the principle, the cooperative repair schemes are presented  in Sections \ref{Divisible} and \ref{Non_Divisible} for $(h+1)\mid 2^n$ and $(h+1)\nmid 2^n$, respectively. Finally, the concluding remark is drawn in Section \ref{Conclusion}.

\section{Preliminaries}\label{Preliminary}
In this section, we briefly review Hadamard MSR codes and Hamming codes. For ease of reading,  we firstly introduce some useful notation used throughout this  paper.

\begin{itemize}
\item Let $\mathbb{F}_q$ be a finite field with $q$ elements, where $q$ is a prime power.
\item Denote $\mathbf{a}\in\mathbb{F}_q^N$ as a vector of length $N$ over $\mathbb{F}_q$.
\item For two non-negative integers $a$ and $b$ with $a < b$, define $[a,b)$ and $[a,b]$ as two sets $\{a,a+1,\cdots,b-1\}$ and $\{a,a+1,\cdots,b\}$, respectively.
\item For any non-negative integer $a\in[0,2^n)$, let $(a_0,a_1,\cdots,a_{n-1})$ be its  binary representation  in  vector form of length $n$, i.e., $a=\sum_{i=0}^{n-1}2^ia_i,a_i\in\{0,1\}$. For convenience, we use both of them alternatively with a slight abuse of notation. Since the length $n$  is determined by the maximum value of $a$, we always specify the range of $a$ before using the binary vector representation.
\item Let $\mathbf{e}_{n,i}$ be a binary vector of length $n$ with only the $i$-th, $i\in[0,n)$ component being non-zero.
\item Let $\mathbf{b}$ and $\mathbf{c}$ be binary vectors of length $n$, denote $\mathbf{b}\oplus\mathbf{c}=(b_0\oplus c_0,b_1\oplus c_1,\cdots,b_{n-1}\oplus c_{n-1})$, where $\oplus$ is the addition operation modulo $2$.
\end{itemize}

\subsection{Hadamard MSR code}

Assume that the original data is of size $M=kN$. An $(n,k)$ MDS code partitions the data into $k$ parts and then encodes into $n$ parts $\mathbf{f}=[\mathbf{f}_0^\top,\mathbf{f}_1^\top,\cdots,\mathbf{f}_{n-1}^\top]^\top$ stored on $n$ nodes, where $\mathbf{f}_i=(f_{i,0},f_{i,1},\cdots,f_{i,N-1})^\top\in \mathbb{F}_q^N,i\in [0,n)$ is a column vector and $\top$ denotes the transpose operator.

 Let $N=2^n$.  The $(n,k)$ Hadamard MSR code is a class of MDS code defined by the following parity-check equations\cite{Hadamard-MSR, Tang-Hadamard}
 \begin{eqnarray*}\label{Parity-check equation}
 A_{t,0}\mathbf{f}_0+ A_{t,1}\mathbf{f}_1+\cdots+ A_{t,n-1}\mathbf{f}_{n-1}=0,\qquad t\in[0,r),
 \end{eqnarray*}
where $A_{t,i}$ is an $2^n\times 2^n$ nonsingular matrix over $\mathbb{F}_q$, called the \textit{parity matrix} of node $i\in[0,n)$
for the $t$-th parity-check equation. In matrix form, the structure
of $(n,k)$ MSR codes based on the above parity-check
equations can be rewritten as
\begin{eqnarray*}\label{Matrix form}
\setlength{\arraycolsep}{0.1pt}
\underbrace{\left(\begin{array}{ccccc}
A_{0,0} & A_{0,1} & \cdots  &   A_{0,n-1}      \\
A_{1,0} & A_{1,1} & \cdots  &   A_{1,n-1}      \\
\vdots  & \vdots  &  \ddots  &  \vdots          \\
A_{r-1,0} & A_{r-1,1} & \cdots  &   A_{r-1,n-1}      \\
\end{array}\right)}_{\mathrm{block ~matrix~}  A}
\left(\begin{array}{ccccc}
\mathbf{f}_0        \\
\mathbf{f}_1\\
\vdots\\
\mathbf{f}_{n-1}\\
\end{array}
\right)=0.
\end{eqnarray*}

Usually, $A$ is designed as a block Vandermonde matrix, i.e.,
\begin{eqnarray*}\label{A_i}
A_{t,i}=A_i^t,\qquad i\in[0,n), ~t\in[0,r),
\end{eqnarray*}
where $A_i,i\in[0,n)$ are $2^n \times 2^n$ nonsingular matrices. In particular, we use the convention $A_i^0=I$.

In fact, Hadamard  MSR  code can be characterized by the coding matrices below,
\begin{eqnarray}\label{Parity-matrix}
A_i=I_{2^{n-i-1}}\otimes \mathrm{blkdiag}(\lambda_{i,0}I_{2^{i}},\lambda_{i,1}I_{2^{i}}),\qquad i\in [0,n),
\end{eqnarray}
where $\otimes$ is Kronecker product, $\lambda_{i,0}, \lambda_{i,1}$ are two distinct elements in $\mathbb{F}_q(q\ge 2n)$.

Accordingly, the $t$-th parity-check equation is
\begin{eqnarray}\label{Eqn_Parity_eq}
\sum\limits_{i=0}^{n-1} \lambda_{i,a}^tf_{i,a}=0, \qquad a\in [0,2^n),~ t\in[0,r).
\end{eqnarray}
From \eqref{Parity-matrix}, it is clear that the diagonal element in row $a\in [0,2^n)$-th of $A_i,i\in [0,n)$ satisfies
\begin{eqnarray}\label{Eqn_lambda}
\lambda_{i,a}=\lambda_{i,a_i}.
\end{eqnarray}
Besides,  the MDS property of Hadamard MSR code further requires for $ i\ne j\in [0,n),a\ne b\in [0,2^n)$,
\begin{eqnarray}\label{Node_lambda}
\lambda_{i,a}\ne \lambda_{j,b}.
\end{eqnarray}

\begin{Example}\label{Example_1}
For $N=2^{14}$, let $\alpha$ be the primitive element of $\mathbb{F}_{29}$. The $(n=14,k=2)$ Hadamard MSR code has
the following coding matrices over $\mathbb{F}_{29}$
\begin{eqnarray*}
A_0&=&\mathrm{diag}(1, -1, 1, -1, 1, -1, 1,-1,\cdots),\\
A_1&=&\alpha \cdot\mathrm{diag}(1, 1, -1, -1, 1, 1, -1,-1, \cdots),\\
A_2&=&\alpha^2\cdot\mathrm{diag}(1, 1, 1, 1, -1, -1, -1,-1, \cdots),\\
 &\vdots&  \\
A_{13}&=&\alpha^{13}\cdot\mathrm{diag}(\underbrace{1, 1, \cdots, 1}_{2^{13}},  \underbrace{-1, -1,\cdots, -1}_{2^{13}})\\
\end{eqnarray*}
  according to \eqref{Parity-matrix} and \eqref{Eqn_lambda}.
\end{Example}

\begin{Remark}
It is known from  \cite{Hadamard-MSR} that the sub-packetization of  $(n,k)$ Hadamard MSR code  is $N=s^n$ with $s=d-k+1$. We focus on the case of $d=k+1$ helper nodes. Thus, only $N=2^n $ is considered in this paper.
\end{Remark}

\subsection{Hamming code}\label{Hamming code}

\begin{Definition}[\textbf{Perfect code} \cite{T-K-Moon}]
A $q$-ary $(n',k')$ code $\mathcal{C}$  is said to be a \textbf{perfect code} if
\begin{eqnarray*}\label{Perfect code}
		q^{n'-k'}= \sum_{i=0}^{\lfloor {d'-1\over 2}\rfloor}{n' \choose i}(q-1)^i,
\end{eqnarray*}
where  $d'$ is the minimum Hamming distance of code $\mathcal{C}$.
\end{Definition}

The sets of perfect codes are actually quite limited, where Hamming codes and Golay codes are the only nontrivial ones \cite{T-K-Moon}. To be specific, given an integer $m\ge 2$, the Hamming  code  $\mathcal{C}$ is a family of $(n'=2^m-1,k'=2^m-1-m)$  binary codes  with single-error-correcting ability, i.e., $d'=3$. Let $V_0=\{\mathbf{c}_0,\mathbf{c}_1,\cdots,\mathbf{c}_{2^{k'}-1}\}$, where $\mathbf{c}_j, 0\le j<2^{k'}$ is the codewords of $\mathcal{C}$.  Then, based on $V_0$, we can define $n'$ sets
\begin{eqnarray}\label{Basic group}
	V_i=\{\mathbf{c}_j\oplus \mathbf{e}_{n',i-1}:\mathbf{c}_j\in V_0,0\le j< 2^{k'}\}, \qquad i\in [1,n'],
\end{eqnarray}
where $\mathbf{e}_{n',i-1}$ is the binary vector of length $n'$ with only the $(i-1)$-th, $i\in[1,n']$ component being non-zero.

The following lemma is a direct consequence of the fact that the binary Hamming code is a perfect code for correcting a single error.

\begin{Lemma}\label{Lem_cover}
The sets $V_0$ and $V_{i},i\in[1,n']$ defined in \eqref{Basic group} can cover all the $2^{n'}$ vectors of length $n'$, i.e.,
\begin{eqnarray*}\label{cup cover}
V_0 \cup V_{1}\cup V_{2}\cup \cdots\cup V_{n'}=[0,2^{n'}).
\end{eqnarray*}
\end{Lemma}

\begin{Example}\label{Hamming code example}
When $m=2$, $(3,1)$-Hamming code $\mathcal{C}$ with the following parity check matrix
\begin{eqnarray*}
H=\left(\begin{array}{cccccccc}
1 & 0 & 1 \\
0 & 1 & 1
\end{array}\right)
\end{eqnarray*}
has  two codewords $\mathbf{c}_0=(000)$ and $\mathbf{c}_1=(111)$.
Table \ref{Hamming code decoding} lists $V_0$ and $V_i,i\in[1,3]$ in \eqref{Basic group}.  In fact, it is just the standard array of $\mathcal{C}$ defined in \cite{T-K-Moon}.

\begin{table}[h]
	\begin{center}\caption{Set $V_{i}$, $i\in[0,3)$}\label{Hamming code decoding}
		\renewcommand\arraystretch{0.8}
		\begin{tabular}{|c|c|c|}
			\hline
			$V_{0}$& $0$ &  $7$ \\
			\hline
			$V_{1}$& $1$ &  $6$   \\
			\hline
			$V_{2}$& $2$ &  $5$   \\
			\hline
			$V_{3}$ &$4$ & $3$   \\
			\hline
		\end{tabular}
	\end{center}
\end{table}
\vspace{1mm}
\end{Example}

\section{The repair principle of Hadamard MSR codes with $d=k+1$}\label{basic repair}

Throughout this paper, it is supposed that  $h$ nodes of an $(n,k)$ Hadamard MSR code $\mathcal{C}$ fail, denoted as $\mathcal{E}=\{i_0,\cdots,i_{h-1}\}$. During repair process of the $h$ failed nodes, $d=k+1$ helper nodes are connected and thus $n-(k+1)-h=r-h-1$ nodes are unconnected,
denoted by $\mathcal{H}=\{j_0,\cdots,j_{k}\}\subseteq[0,n)\backslash \mathcal{E},|\mathcal{H}|=k+1$ and  $\mathcal{U}=[0,n)\backslash(\mathcal{E}\cup\mathcal{H})=\{z_0,\cdots,z_{r-h-2}\}$, respectively.

In this paper, we have the following \textbf{\textit{repair principle}} for the Hadamard MSR code with $h$ failed nodes.
\begin{itemize}
	\item  \textbf{\textit{Grouping}:} ~Divide indices of the stored symbols into $M$ groups according to some certain rules (See Sections \ref{Divisible} and  \ref{Non_Divisible} for the specific  rules), namely $S_0,\cdots,S_{M-1}$.

Denote  $S_g=\{S_g(0),\cdots,S_g(\frac{N}{M}-1)\}$, $0\le g<M$, where $S_g(v),v\in[0,\frac{N}{M})$ is called the $v$-th element of group $g$. Then, the stored symbol $\mathbf{f}_i=(f_{i,0},f_{i,1},\cdots,f_{i,N-1})^\top$ can be rewritten as $\mathbf{f}_i=(f_{i,S_g(v)})^\top_{g\in [0,M),v\in[0,\frac{N}{M})},i\in [0,n)$.
	
	\item  \textbf{\textit{Pairing}:} ~Given an index $g$,  pair it with a $g'\in[0,M)$ to repair node $i_u\in\mathcal{E}$ (Refer to Sections \ref{Divisible} and  \ref{Non_Divisible} for the specific pairing strategy).
	
	\item  \textbf{\textit{Downloading}:} ~Node $i_u$ downloads
	\begin{eqnarray}\label{Eqn_u_download}
		f_{j,S_g(v)}+f_{j,S_{g'}(v)},\qquad v\in[0,\frac{N}{M})
	\end{eqnarray}
from all the helper nodes $j\in\mathcal{H}$.
\end{itemize}

\begin{Theorem}\label{Thm_Repair Condition}
	According to \textbf{Pairing} of \textbf{the repair principle}, for the given $g,g'\in [0,M)$, by downloading $f_{j,S_g(v)}+f_{j,S_{g'}(v)}, v\in[0,\frac{N}{M}),j\in\mathcal{H}$, node $i_u\in \mathcal{E}$ can recover
	\begin{itemize}
		\item $f_{i_u,S_g(v)}$,$f_{i_u,S_{g'}(v)}$, $v\in[0,\frac{N}{M})$ for itself, and
		\item $f_{i,S_g(v)}+f_{i,S_{g'}(v)}$, $v\in[0,\frac{N}{M})$ for other failed nodes $i\in \mathcal{E}\backslash\{i_u\}$,
	\end{itemize}
if
\begin{numcases}{}
	\lambda_{i_u,S_g(v)}\neq \lambda_{i_u,S_{g'}(v)},\label{Condition-0}\\
	\lambda_{i,S_g(v)}=\lambda_{i,S_{g'}(v)},\qquad\forall i\in [0,n)\backslash\{i_u\}.\label{Condition-1}
\end{numcases}
\end{Theorem}

\begin{proof}
	It follows from \eqref{Eqn_Parity_eq} that for any $v\in [0,\frac{N}{M})$ and $t\in[0,r)$,
	\begin{eqnarray}\label{G_g}
		\begin{array}{l}
			\lambda_{i_u,S_g(v)}^tf_{i_u,S_g(v)}+
			\sum\limits_{i\in \mathcal{E}\backslash\{i_u\}} \lambda_{i,S_g(v)}^tf_{i,S_g(v)}+
			\sum\limits_{j\in\mathcal{H}}\lambda_{j,S_g(v)}^tf_{j,S_g(v)}
			+\sum\limits_{z\in\mathcal{U}}\lambda_{z,S_g(v)}^tf_{z,S_g(v)}=0
		\end{array}
	\end{eqnarray}
	and
	\begin{eqnarray}\label{G_g'}
		\begin{array}{l}
			\lambda_{i_u,S_{g'}(v)}^tf_{i_u,S_{g'}(v)}+
			\sum\limits_{i\in \mathcal{E}\backslash\{i_u\}} \lambda_{i,S_{g'}(v)}^tf_{i,S_{g'}(v)}+
			\sum\limits_{j\in\mathcal{H}}\lambda_{j,S_{g'}(v)}^tf_{j,S_{g'}(v)}
			+\sum\limits_{z\in\mathcal{U}}\lambda_{z,S_{g'}(v)}^tf_{z,S_{g'}(v)}=0.
		\end{array}
	\end{eqnarray}
	
	Substituting \eqref{Condition-0} and \eqref{Condition-1} to the summation of $\eqref{G_g}$ and $\eqref{G_g'}$, we get
	\begin{eqnarray}\label{Eqn_Gg+g'}
		&&\lambda_{i_u,S_g(v)}^tf_{i_u,S_g(v)}+\lambda_{i_u,S_{g'}(v)}^tf_{i_u,S_{g'}(v)}
		+\sum\limits_{i\in \mathcal{E}\backslash\{i_u\}} \lambda_{i,S_g(v)}^t(f_{i,S_g(v)}+f_{i,S_{g'}(v)})
		+\sum\limits_{z\in\mathcal{U}}\lambda_{z,S_g(v)}^t(f_{z,S_g(v)}+f_{z,S_{g'}(v)})\nonumber\\
		&=&-\sum\limits_{j\in\mathcal{H}}\lambda_{j,S_g(v)}^t(f_{j,S_g(v)}+f_{j,S_{g'}(v)}).
	\end{eqnarray}
	
	For a fixed $v\in [0,\frac{N}{M})$, since we have downloaded $f_{j,S_g(v)}+f_{j,S_{g'}(v)}$ in \eqref{Eqn_u_download} from all the helper nodes $j\in\mathcal{H}$, the data on the right-hand side of \eqref{Eqn_Gg+g'} is known. Then, there are $r$ unknowns $f_{i_u,S_g(v)},f_{i_u,S_{g'}(v)},f_{i,S_g(v)}+f_{i,S_{g'}(v)}, i\in \mathcal{E}\backslash\{i_u\}$ and
	$f_{z,S_g(v)}+f_{z,S_{g'}(v)}, z\in\mathcal{U}$ on the left-hand side of \eqref{Eqn_Gg+g'}. When $t$ enumerates $[0,r)$, by means of reorder of  columns, the $r\times r$ coefficient matrix on the left of the equation \eqref{Eqn_Gg+g'} can be written as the following Vandermonde matrix
\begin{eqnarray*}\label{Parity matrix}
		A=\left(\begin{array}{ccccccccccc}
			1  & 1 & \cdots  & 1 & 1 & \cdots   & 1 \\
			\lambda_{i_u,S_{g'}(v)} & \lambda_{i_0,S_g(v)} & \cdots & \lambda_{i_{h-1},S_g(v)} & \lambda_{z_0,S_g(v)}& \cdots & \lambda_{z_{r-h-2},S_g(v)} \\
			\vdots & \vdots & \ddots & \vdots & \vdots & \ddots & \vdots \\
			\lambda_{i_u,S_{g'}(v)}^{r-1} & \lambda_{i_0,S_g(v)}^{r-1} & \cdots  & \lambda_{i_{h-1},S_g(v)}^{r-1} & \lambda_{z_0,S_g(v)}^{r-1}& \cdots & \lambda_{z_{r-h-2},S_g(v)}^{r-1}  \\
		\end{array}\right)
	\end{eqnarray*}
which  is invertible  due to \eqref{Node_lambda} and \eqref{Condition-0}. Thus, with  $v$ ranging over $[0,\frac{N}{M})$, node $i_u$ gets the desired data as claimed.
\end{proof}

\begin{Example}\label{Example_3}
	Assume that nodes $0,1,2$ fail in Example \ref{Example_1}, the helper nodes and unconnected nodes are $[3,5]$ and $[6,13]$, respectively. Take  symbols
	$0$ and $1$ of the node $0$ as an example.
	
	For $t\in[0,11]$, from \eqref{Eqn_Parity_eq}  we know
	\begin{eqnarray*}\label{1_0}
		\begin{array}{l}	f_{0,0}+\alpha^tf_{1,0}+\alpha^{2t}f_{2,0}+\sum_{j=3}^{13}\alpha^{jt}f_{j,0}=0
		\end{array}
	\end{eqnarray*}
	and
	\begin{eqnarray*}\label{1_1}
		\begin{array}{l} (-1)^tf_{0,1}+\alpha^tf_{1,1}+\alpha^{2t}f_{2,1}+\sum_{j=3}^{13}\alpha^{jt}f_{j,1}=0,
		\end{array}
	\end{eqnarray*}
which gives
	\begin{eqnarray}\label{1_0+1_1}
		f_{0,0}+(-1)^tf_{0,1}+\alpha^t(f_{1,0}+f_{1,1})+\alpha^{2t}(f_{2,0}+f_{2,1})+\sum_{z=6}^{13}\alpha^{zt}(f_{z,0}+f_{z,1})
		&=&-\sum_{j=3}^{5}\alpha^{jt}(f_{j,0}+f_{j,1}).
	\end{eqnarray}
	
	Node 0 downloads $f_{j,0}+f_{j,1}$ from helper nodes $j\in[3,5]$, so the data on the right-hand side of equation \eqref{1_0+1_1} is known. We can get the reordered coefficient matrix on the left-hand side
	\begin{eqnarray*}
		A=\left(\begin{array}{ccccccccccc}
			1      & 1       & 1  & 1           & 1         & \cdots & 1 \\
			-1     & 1       & \alpha   & \alpha^2    & \alpha^6  & \cdots & \alpha^{13} \\
			\vdots & \vdots  & \vdots   & \vdots      & \vdots    & \ddots & \vdots \\
			-1     & 1       & \alpha^{11}  & \alpha^{22} & \alpha^{10}    & \cdots & \alpha^{3}  \\
		\end{array}\right).
	\end{eqnarray*}

Then, the equation \eqref{1_0+1_1} has unique solution because of $Rank(A)=12$. Hence, node $0$ can recover symbols $f_{0,0}$ and $f_{0,1}$ for itself and $f_{1,0}+f_{1,1},f_{2,0}+f_{2,1}$ for  nodes 1 and 2. 	
\end{Example}

Noting that the above repair principle was first proposed in  \cite{New-Constructions}, we generalize it here. Clearly, the key point of the repair  principle is how to make group and pair such that the conditions \eqref{Condition-0} and \eqref{Condition-1} can be satisfied. In \cite{New-Constructions}, grouping and pairing are based on the instances generated by space sharing technique.  Whereas, we will give different method based on Hamming codes in the next two sections, which may result in smaller sub-packetization.

\section{Optimal cooperative repair scheme for Hadamard MSR codes with $(h+1)|2^n$}\label{Divisible}

In this section, we propose an optimal  cooperative repair scheme for $(n,k)$ Hadamard MSR code $\mathcal{C}$ with sub-packetization $N=2^n$ when the number of failed nodes satisfies $h\ge 2$ and $(h+1)|2^n$. In this case, $h=2^m-1$ for  a positive integer $2\le m \le \log_2(n-k+1)$. $\mathcal{E}=\{i_0,\cdots,i_{h-1}\}$ and $\mathcal{H}=\{j_0,\cdots,j_{k}\}\subseteq[0,n)\backslash \mathcal{E}$ are the sets of $h$ failed nodes and $k+1$ helper nodes, respectively.

\textbf{\textit{Grouping}:} ~Following \textbf{repair principle} in the last section, we first divide the symbols at each node of Hadamard MSR code $\mathcal{C}$ into $M=2^h$ groups, each group having $2^{n-h}$ symbols. For given $g=(g_0,\cdots,g_{h-1})\in[0,2^h)$, define an ordered indices set as
\begin{eqnarray}\label{Eqn_S_g_11}
	S_g &=&\{a:a_{i_u}=g_u,a=(a_0,\cdots, a_{n-1})\in[0,2^n),u\in[0,h)\}\nonumber\\
	&=&\{S_g(0),\cdots,S_g(2^{n-h}-1)\},
\end{eqnarray}
where
\begin{eqnarray}\label{Eqn_S_g_12}
	S_g(0)<S_g(1)<\cdots<S_g(2^{n-h}-1).
\end{eqnarray}

Consequently, we have $\mathbf{f}_i=(f_{i,S_g(v)})^\top_{g\in [0,2^h),v\in[0,2^{n-h})}$, $i\in[0,n)$.

\textbf{\textit{Pairing}:} ~Define $V_0$ and $V_i, i\in [1,h]$ by applying $(h=2^m-1, 2^m-m-1)$ Hamming code to  \eqref{Basic group} in place of $n'=h$ and
$k'=2^m-m-1$. For $g=(g_0,\cdots,g_{h-1})\in V_0$, pair it with the index $g'=g\oplus \mathbf{e}_{h,u}\in V_{u+1}$ to repair
node $i_u\in \mathcal{E}$.

\begin{Lemma}\label{Lem_Inter-groups}
	Given a failed node $i_u,u\in[0,h)$, for any $g,g'\in [0,2^h)$, \eqref{Condition-0} and \eqref{Condition-1} are satisfied if $g'=g\oplus \mathbf{e}_{h,u}$.
\end{Lemma}

\begin{proof}
	For any $v\in[0,2^{n-h})$, let $a=S_g(v)\in [0,2^n), b=S_{g'}(v)\in [0,2^n)$. From \eqref{Eqn_S_g_11}, we know
	$$a_{i_u}=g_u,b_{i_u}=g'_u.$$
	When $g'=g\oplus \mathbf{e}_{h,u}$, we can obtain
	$$g'_u=g_u\oplus 1.$$
	That is,
	$$b_{i_u}=a_{i_u}\oplus 1$$
and \begin{equation*}
S_{g'}=
\begin{cases}
S_g+2^{i_u}=\{\tilde{a}+2^{i_u}:\tilde a\in S_g\} &\qquad g_u=0,\\
S_g-2^{i_u}=\{\tilde{a}-2^{i_u}:\tilde a\in S_g\} &\qquad g_u=1.\\
\end{cases}
\end{equation*}
	
	Associated with the ordering of \eqref{Eqn_S_g_12}, for a given $v\in[0,2^{n-h})$,
$$b=S_{g'}(v)=\left\{\begin{array}{ll}
	S_g(v)+2^{i_u}=a+2^{i_u}, & \mathrm{if}~a_{i_u}=0,\\
	S_g(v)-2^{i_u}=a-2^{i_u}, & \mathrm{if}~a_{i_u}=1,
\end{array}
\right.$$
which means
	\begin{eqnarray*}
		b_{i_u}&=&a_{i_u}\oplus 1,\\
		b_{i}&=&a_{i}, \qquad i\in [0,n)\setminus\{i_u\}.
	\end{eqnarray*}
The preceding equations imply \eqref{Condition-0} and \eqref{Condition-1} by \eqref{Eqn_lambda}, which completes the proof.
\end{proof}

When $(h+1)|2^n$,  the cooperative repair scheme of $h$ failed nodes is given as follows.

\textbf{Download  phase:} ~For any $g\in V_0$ and $g'=g\oplus \mathbf{e}_{h,u}$, failed node $i_u\in\mathcal{E}$ downloads $f_{j,S_g(v)}+f_{j,S_{g'}(v)}, v\in 2^{n-h}$ defined in \eqref{Eqn_u_download} from helper nodes $j\in \mathcal{H}$ to recover \begin{eqnarray}\label{Eqn_iu_recover_11}
	\{f_{i_u,S_g(v)},f_{i_u,S_{g\oplus \mathbf{e}_{h,u}}(v)}:g\in V_0,v\in[0,2^{n-h})\}
\end{eqnarray}
for itself and
\begin{eqnarray}\label{Eqn_iu_recover_12}
	\{f_{i,S_g(v)}+f_{i,S_{g\oplus \mathbf{e}_{h,u}}(v)}:g\in V_0,v\in[0,2^{n-h})\}
\end{eqnarray}
for other failed nodes $i\in \mathcal{E}\backslash\{i_u\}$ according to Lemma  \ref{Lem_Inter-groups} and Theorem  \ref{Thm_Repair Condition}.

\vspace{3mm}

\textbf{Cooperative phase:} ~Other failed node $i_{\bar{u}},\bar{u}\in [0,h)\setminus\{u\}$ transfers $\{f_{i_u,S_g(v)}+f_{i_u,S_{g\oplus\mathbf{e}_{h,\bar{u}}}(v)}:g\in V_0,
v\in[0,2^{n-h})\}$  recovered in \eqref{Eqn_iu_recover_12} at download phase to node $i_u$.

Node $i_u$ utilizes its own data $\{f_{i_u,S_g(v)},g\in V_0,v\in[0,2^{n-h})\}$ recovered in \eqref{Eqn_iu_recover_11}  to solve
\begin{eqnarray}\label{Eqn_iu_recover_13}
	&&\{f_{i_u,S_{g\oplus\mathbf{e}_{h,\bar{u}}}(v)}:g\in V_0,
	v\in[0,2^{n-h}),\qquad \bar{u}\in [0,h)\setminus\{u\}\}
\end{eqnarray}
from $\{f_{i_u,S_g(v)}+f_{i_u,S_{g\oplus\mathbf{e}_{h,\bar{u}}}(v)}:g\in V_0,v\in[0,2^{n-h}),\bar{u}\in [0,h)\setminus\{u\}\}$ .

Combining \eqref{Eqn_iu_recover_11} and \eqref{Eqn_iu_recover_13},  node $i_u$ obtains
\begin{eqnarray}\label{Eqn_iu_recover_14}
	\{f_{i_u,S_g(v)}, f_{i_u,S_{g\oplus\mathbf{e}_{h,i}}(v)}:g\in V_0,i\in [0,h),
	v\in[0,2^{n-h})\}.
\end{eqnarray}

From \eqref{Basic group}, $\{g\oplus\mathbf{e}_{h,i}: g\in V_0\}$ enumerates $V_1,\cdots, V_h$ with  $i$ ranging over $[0,h)$.
According to Lemma
\ref{Lem_cover},
\begin{eqnarray*}\label{cup cover}
V_0 \cup V_{1}\cup V_{2}\cup \cdots\cup V_{h}=[0,2^h).
\end{eqnarray*}
Then,  as $v$ enumerating $[0,2^{n-h})$,  by \eqref{Eqn_S_g_11} we have
\begin{eqnarray}\label{Eqn_iu_rang}
	\bigcup_{g\in V_0,
		v\in[0,2^{n-h}),i\in [0,h)} \{S_g(v),S_{g\oplus\mathbf{e}_{h,i}}(v)\} = [0,2^n).
\end{eqnarray}
Thus, associated with \eqref{Eqn_iu_rang}, equation \eqref{Eqn_iu_recover_14} indicates that all data $\{f_{i_u,0},\cdots,f_{i_u,2^n-1}\}$ of the failed node $i_u\in\mathcal{E}$ has been recovered.

During the download  phase, each failed node downloads $2^{h-m}\cdot 2^{n-h}=2^{n-m}=N/(h+1)$ symbols from each of the $k+1$ helper nodes, i.e.,  the repair bandwidth at this phase is
\begin{eqnarray*}
	\gamma_1 &= h(k+1)\cdot{N\over h+1}.
\end{eqnarray*}

Next, during the cooperative phase, each failed node accesses $2^{h-m}\cdot 2^{n-h}=2^{n-m}=N/(h+1)$ symbols from other $h-1$ failed nodes, i.e., the repair bandwidth at this phase is
\begin{eqnarray*}\gamma_2 &= h(h-1)\cdot{N\over h+1}.
\end{eqnarray*}
Totally, the repair bandwidth is
\begin{eqnarray*}\label{Eqn_Repair Bandwidth1}
	\gamma &= & \gamma_1+\gamma_2 \nonumber\\
	&=&\frac{h(k+h)\times N}{h+1}
\end{eqnarray*}
which attains the optimal repair bandwidth according to \eqref{Eqn_Lower_Bound}.

This concludes the property of our repair scheme as well as the proof of the following Theorem \ref{Thm_Repair_1}.
\begin{Theorem}\label{Thm_Repair_1}
	When $(h+1)|2^n$, any $h\ge 2 $ failed nodes of $(n,k)$ Hadamard MSR code $\mathcal{C}$ with sub-packetization $N=2^n$ can be  optimally cooperative repaired with $d=k+1$ helper nodes.
\end{Theorem}

\begin{Remark}
For $(h+1)|2^n$, the  optimal cooperative repair scheme in  \cite{New-Constructions} requires  the sub-packetization $(h+1)N$. In contrast, the sub-packetization in Theorem \ref{Thm_Repair_1} is  $h+1=2^m$ times less.
\end{Remark}

\begin{Example} Continued with Example  \ref{Example_3}. Firstly, divide $N=2^{14}$ symbols into $M=8$ groups as follows, each group having $2^{11}$ symbols.
		\begin{eqnarray*}
			S_0&=&\{a:a_0=0,a_1=0,a_2=0,a\in[0,2^{14})\},\\
			S_1&=&\{a:a_0=1,a_1=0,a_2=0,a\in[0,2^{14})\},\\
			&\vdots &\\
			S_7&=&\{a:a_0=1,a_1=1,a_2=1,a\in[0,2^{14})\}.
		\end{eqnarray*}	

The binary representation  of these groups can be regarded as the codewords of $(3,1)$-Hamming code, so they can be partitioned into $V_0,V_1,V_2$ and $V_3$ as given in  Example \ref{Hamming code example}.

The repairing processes for failed nodes $\{0,1,2\}$ are explained in Figure \ref{fig:repairing}. Tables  \ref{Divisible-group_1} and  \ref{Divisible-group_2} illustrate the downloaded and recovered data at download phase and cooperative phase, respectively. The pairing $(V_0,V_j),j\in[1,3]$ means that all the groups $g\in V_0$ are paired with $g'=g\oplus \mathbf{e}_{3,j-1}\in V_j$, one by one, to repair node $j-1$.
\begin{figure*}[tb]
  \centering
\tikzsetnextfilename{2-butterflies}
\scalefont{0.5}
\begin{tikzpicture}[scale =1]
    \node[mybox] (nodein1+) { node 0} ;
	\node[mybox, above =20pt  of nodein1+] (nodein2+) { node 1} ;
    \node[mybox, above =20pt of nodein2+] (nodein3+) { node 2} ;
    \node[mybox, above =20pt of nodein3+] (nodein4+) { node 3} ;
    \node[mybox, above =20pt of nodein4+] (nodein5+) { node 4} ;
    \node[mybox, above =20pt of nodein5+] (nodein6+) { node 5} ;

    \node[mybox, right =80 pt of nodein3+] (node1+) {$\ \ $} ;

	\node[mybox, right =80pt  of nodein4+] (node2+) {$\ \ $} ;
    \node[mybox, right =80pt of nodein5+] (node3+) {$\ \ $} ;
    \node[mybox, right =80 pt of node1+] (node1) {$\ \ $} ;
	\node[mybox, right =80pt  of node2+] (node2) {$\ \ $} ;
    \node[mybox, right =80pt of node3+] (node3) {$\ \ $} ;
    \node[mybox, right =90 pt of node1] (node1new) { node 0} ;
	\node[mybox, right =90pt  of node2] (node2new) { node 1} ;
    \node[mybox, right =90pt of node3] (node3new) { node 2} ;
    \node[rectangle,above =3pt of nodein6+](text1){$\vdots$};
    \node[mybox, above =3pt of text1] (nodein14+) { node 13} ;

	\path[black,->,very thick,auto]
          (nodein4+.east) edge [bend right] node {$\color{red}D_{3,2}$} (node3+.west)
          (nodein5+.east) edge  node {${\color{red}D_{4,2}}$} (node3+.west)
          (nodein6+.east) edge [bend left] node {$\color{red}D_{5,2}$} (node3+.west)
          (node1.east) edge [bend right] node {} (node3.east)
          (node2.east) edge [bend right] node {$\color{red}CO_{1,2}$} (node3.east)
          (8.4,3) edge node {} (node1.east)
          (8.4,2.7) edge node {} (node1.east)
          (8.4,4.5) edge node {} (node2.east)
          (8.4,4.2) edge node {} (node2.east)
          (2.4,3.5) edge node {} (node1+.west)
          (2.4,3) edge node {} (node1+.west)
          (2.4,2.5) edge node {} (node1+.west)
          (2.4,5) edge node {} (node2+.west)
          (2.4,4.5) edge node {} (node2+.west)
          (2.4,4) edge node {} (node2+.west);
    \draw[pattern= dots] (6.75,3.96) rectangle (7.43,4.3);
    \draw[pattern= dots] (6.75,5.4) rectangle (7.43,5.74);
    \draw[pattern= dots] (6.75,2.52) rectangle (7.43,2.86);
	\draw [|-|, very thick] (2.8,2.5) arc  (240:120:15pt);
    \draw node[] at (2.8,2.3) {$\color{red}D_{0}$};

    \draw node[] at (3.2,1.5) {Download Phase};

    \draw [|-|, very thick] (2.8,4) arc  (240:120:15pt);
    \draw node[] at (2.8,3.8) {$\color{red}D_{1}$};

    \draw [|-|, very thick] (8.2,4) arc  (-40:40:15pt);
    \draw node[] at (8.2,3.8) {$\color{red}CO_{1}$};
    \draw [|-|, very thick] (8.2,2.5) arc  (-40:40:15pt);
    \draw node[] at (8.2,2.3) {$\color{red}CO_{0}$};

   \draw [->,thick,dotted](4.4,4.4) -- (6.4,4.4);
   \draw [->,thick,dotted](8.4,4.4) -- (10.4,4.4);
   \draw node[] at (5.4,4.1) {recovering};
   \draw node[] at (9.4,4.1) {recovering};
   \draw node[] at (7.4,3.5) {$\color{red}CO_{0,2}$};
   \draw node[] at (7.4,1.5) {Cooperative Phase};
   \draw [-,thick](0.25,0.25) -- (-0.25,-0.25);
   \draw [-,thick](-0.25,0.25) -- (0.25,-0.25);
   \draw [-,thick](0.25,1.65) -- (-0.25,1.15);
   \draw [-,thick](-0.25,1.65) -- (0.25,1.15);
   \draw [-,thick](0.25,3.1) -- (-0.25,2.6);
   \draw [-,thick](-0.25,3.1) -- (0.25,2.6);
\end{tikzpicture}

  \caption{Repairing processes for erasures, where $D_2=\{D_{3,2},D_{4,2},D_{5,2}\}$, $CO_2=\{CO_{0,2},CO_{1,2}\}$, and $D_i$, $CO_i$ denote the data downloaded for failed nodes $0\leq i\leq 2$ in download phase and cooperative phase, respectively. For more specific data downloaded and recovered in each phase, refer to Tables \ref{Divisible-group_1} and
   \ref{Divisible-group_2}.}
  \label{fig:repairing}
\end{figure*}
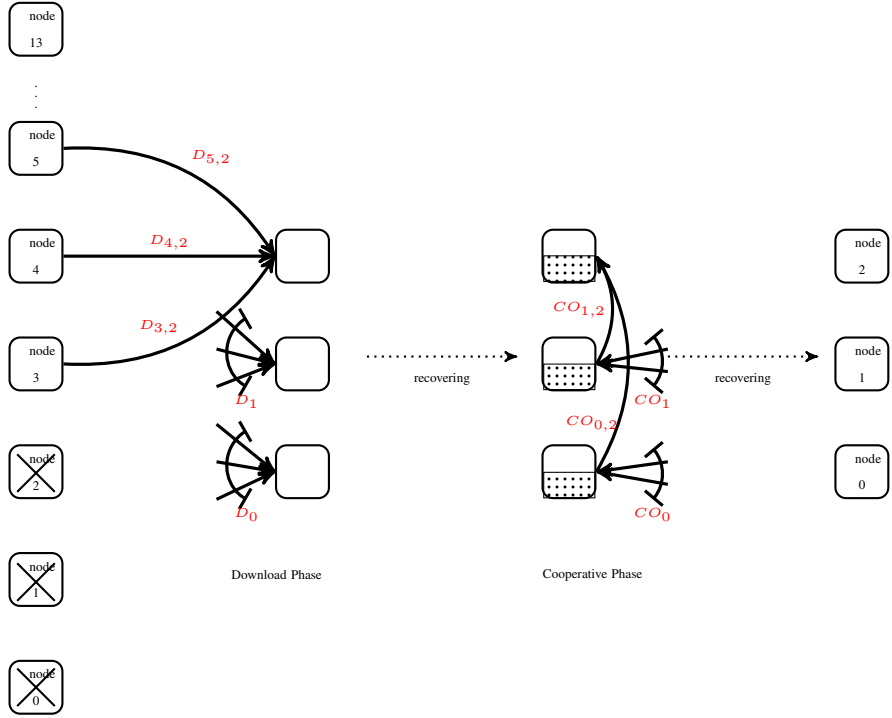

	\begin{table}[h]
		\begin{center}
			\caption{Download phase of repairing failed nodes $\{0,1,2\}$, where $v\in[0,7]$}
			\label{Divisible-group_1}
			\setlength{\tabcolsep}{8pt}
			\begin{tabular}{|c|c|c|c|c|c|c|c|}
				\hline
				Failed nodes & 0 & 1 & 2 \\
				\hline
				Pairing &  $(V_0,V_1)$ & $(V_0,V_2)$ & $(V_0,V_3)$\\
				\hline
				\multirow{2}{*}{Download} &$ f_{j,S_0(v)}+f_{j,S_1(v)}$ &$ f_{j,S_0(v)}+f_{j,S_2(v)}$ &$ f_{j,S_0(v)}+f_{j,S_4(v)}$ \\
				&$ f_{j,S_7(v)}+f_{j,S_6(v)}$ &$ f_{j,S_7(v)}+f_{j,S_5(v)}$ &$ f_{j,S_7(v)}+f_{j,S_3(v)}$\\
				\hline
				\multirow{6}{*}{Repair} &$ f_{0,S_0(v)},f_{0,S_1(v)}$ &$ f_{1,S_0(v)},f_{1,S_2(v)}$ &$ f_{2,S_0(v)},f_{2,S_4(v)}$ \\
				&$ f_{0,S_7(v)},f_{0,S_6(v)}$ &$ f_{1,S_7(v)},f_{1,S_5(v)}$ &$ f_{2,S_7(v)},f_{2,S_3(v)}$\\
				&$ f_{1,S_0(v)}+f_{1,S_1(v)}$ &$ f_{0,S_0(v)}+f_{0,S_2(v)}$ &$ f_{0,S_0(v)}+f_{0,S_4(v)}$ \\
				&$ f_{1,S_7(v)}+f_{1,S_6(v)}$ &$ f_{0,S_7(v)}+f_{0,S_5(v)}$ &$ f_{0,S_7(v)}+f_{0,S_3(v)}$\\
				&$ f_{2,S_0(v)}+f_{2,S_1(v)}$ &$ f_{2,S_0(v)}+f_{2,S_2(v)}$ &$ f_{1,S_0(v)}+f_{1,S_4(v)}$ \\
				&$ f_{2,S_7(v)}+f_{2,S_6(v)}$ &$ f_{2,S_7(v)}+f_{2,S_5(v)}$ &$ f_{1,S_7(v)}+f_{1,S_3(v)}$\\
				\hline
			\end{tabular}
		\end{center}
	\end{table}
	
	\begin{table}[h]
		\begin{center}
			\caption{Cooperative phase of repairing failed nodes $\{0,1,2\}$, where $v\in[0,7]$}
			\label{Divisible-group_2}
			\setlength{\tabcolsep}{8pt}
			\begin{tabular}{|c|c|c|c|c|c|c|c|}
				\hline
				Failed nodes & 0 & 1 & 2 \\
				\hline
				Transfer &  $1,2$ &  $0,2$&  $0,1$		\\
				\hline
				\multirow{4}{*}{Cooperative} &$ f_{0,S_0(v)}+f_{0,S_2(v)}$ & $f_{1,S_0(v)}+f_{1,S_1(v)}$ & $ f_{2,S_0(v)}+f_{2,S_1(v)}$ \\
				&$ f_{0,S_7(v)}+f_{0,S_5(v)}$ & $ f_{1,S_7(v)}+f_{1,S_6(v)}$ & $ f_{2,S_7(v)}+f_{2,S_6(v)}$\\
				&$ f_{0,S_0(v)}+f_{0,S_4(v)}$ & $ f_{1,S_0(v)}+f_{1,S_4(v)}$ & $ f_{2,S_0(v)}+f_{2,S_2(v)}$ \\
				&$ f_{0,S_7(v)}+f_{0,S_3(v)}$ & $ f_{1,S_7(v)}+f_{1,S_3(v)}$ & $ f_{2,S_7(v)}+f_{2,S_5(v)}$ \\
				\hline
				\multirow{4}{*}{Repair} &$f_{0,S_2(v)}$ & $f_{1,S_1(v)}$ & $f_{2,S_1(v)}$ \\
				&$f_{0,S_5(v)}$ & $f_{1,S_6(v)}$ & $f_{2,S_6(v)}$\\
				&$f_{0,S_4(v)}$ & $f_{1,S_4(v)}$ & $f_{2,S_2(v)}$ \\
				&$f_{0,S_3(v)}$ & $f_{1,S_3(v)}$ & $ f_{2,S_5(v)}$ \\
				\hline
			\end{tabular}
		\end{center}
	\end{table}	
\end{Example}

\section{Optimal cooperative repair scheme for Hadamard MSR codes with $(h+1)\nmid2^n$}\label{Non_Divisible}

When $(h+1)\nmid2^n$, set $h=(2\ell+1)2^m-1\le n-k$ for some integers $m\ge 0$ and $\ell\ge 1$. Let $h'=2^m-1$. $\mathcal{E}=\{i_0,\cdots,i_{h-1}\}$ and $\mathcal{H}=\{j_0,\cdots,j_{k}\}\subseteq[0,n)\backslash \mathcal{E}$ are the sets of $h$ failed nodes and $k+1$ helper nodes, respectively. Before presenting feasible pairing method such that \eqref{Condition-0} and \eqref{Condition-1}  are satisfied, we first expand the sub-packetization of  $(n,k)$ Hadamard MSR code $\mathcal{C}$ from $2^n$ to $(2\ell+1)2^n$. Then, we  propose the cooperative repair scheme for the $h$ failed nodes.

\subsection{$(n,k)$ Hadamard MSR code $\mathcal{C}$ with sub-packetization  $N=(2\ell+1)2^n$}\label{Construction}

In the original $(n,k)$ Hadamard MSR code $\mathcal{C}$ with sub-packetization  $2^n$, each node $i\in[0,n)$ stores a column vector  $\mathbf{f}_i=(f_{i,0},f_{i,1},\cdots,$ $f_{i,2^n-1})^\top$ of length $2^n$. Generate $2\ell+1$ instances of $\mathcal{C}$, whose  column vectors are denoted by $\mathbf{f}_i^{(0)}, \cdots,\mathbf{f}_i^{(2\ell)}, 0\le i<n$. We  obtain the desired $(n,k)$ Hadamard MSR code with sub-packetization  $N=(2\ell+1)2^n$. By convenience, still denote the code  by $\mathcal{C}$ and write  the column vector of length $(2\ell+1)2^n$ stored at node $i$ as $\mathbf{f}_i=((\mathbf{f}_i^{(0)})^\top, \cdots,(\mathbf{f}_i^{(2\ell)})^\top)^\top,i\in[0,n)$.

Recall from Section   \ref{Divisible} that the $2^n$ symbols of $\mathbf{f}_i^{(w)},w\in [0,2\ell]$ are divided into $2^{h'}$ groups, where $h'=2^m-1$, each group having $2^{n-h'}$ symbols. That means the data of each nodes is divided into $(2\ell+1)\times 2^{h'}$ groups. Similarly to \eqref{Eqn_S_g_11}, for $g\in [0,2^{h'})$, define $S_{g}$ as
\begin{eqnarray}\label{Eqn_S_g_2}
	S_g &=&\{a:a_{i_{(2\ell+1)u}}=g_u,a=(a_0,\cdots, a_{n-1})\in[0,2^n),u\in \{0,1,\cdots,h'-1\}\}\nonumber\\
	&=&\{S_g(0),\cdots,S_g(2^{n-h'}-1)\},
\end{eqnarray}
where $i_{(2\ell+1)u}\in\mathcal{E}=\{i_0,\cdots,i_{h-1}\}$ and
\begin{eqnarray}\label{order}
	S_g(0)<S_g(1)<\cdots<S_g(2^{n-h'}-1).
\end{eqnarray}

Then,  the parity-check equation \eqref{Eqn_Parity_eq} of code $\mathcal{C}$ with sub-packetization  $N=(2\ell+1)2^n$ can be rewritten as
\begin{eqnarray}\label{Eqn_C_21}
	\sum\limits_{i=0}^{n-1}\lambda_{i,S_g(v)}^tf_{i,S_g(v)}^{(w)}=0, \qquad w\in [0,2\ell],~ g\in[0,2^{h'}),~v\in[0,2^{n-h'}), ~t\in[0,r).
\end{eqnarray}

\subsection{Pairing}

In this subsection, we fix the failed node $i_u,u=(2\ell+1)u_1+u_2\in[0,h),u_1\in[0,h'),u_2\in[0,2\ell]$ or $u_1=h',u_2\in[0,2\ell)$. As the data of each nodes is divided into $(2\ell+1)\times 2^{h'}$ groups, we employ indices $(w,g),w\in [0,2\ell],g\in [0,2^{h'})$  to denote these groups, and denote the $v$-th symbol by $S_{w,g}(v)$, where $v\in[0,2^{n-h'})$.  Naturally, $S_g$ in \eqref{Eqn_S_g_2} is defined for $w=0$, i.e. $S_{0,g} =S_g$. Directly, for any $w\in [0,2\ell]$, we define ordered set
\begin{eqnarray}\label{Eqn_def_s_wg}
	S_{w,g} &=&S_g,\qquad g\in[0,2^{h'})
\end{eqnarray}
with a default order
\begin{eqnarray}\label{S_w,g}
	S_{w,g}(v) &\triangleq&S_g(v), \qquad g\in[0,2^{h'}),~v\in[0,2^{n-h'}).
\end{eqnarray}

Similarly, we need to pair these groups.

\textbf{\textit{Pairing}:} ~In order to make \eqref{Condition-0} and \eqref{Condition-1} valid, we have to reorder some groups. For the sake of simplification, the $v$-th symbol is still denoted as $S_{w,g}(v)$. Here we highlight that the reorder is related with the index of node $i_u$.

\begin{itemize}
	\item When $0\le u_1< h',u_2=0$, the same as the pairing method in section  \ref{Divisible}, pair the group $(0,g)\in \{0\}\times V_0$ with $ (0,g\oplus \mathbf{e}_{h',u_1})\in \{0\} \times V_{u_1+1}$. In this case, we use the default order of the latter, i.e.
	\begin{eqnarray}\label{Eqn_S_wg_1}
		S_{0,g\oplus \mathbf{e}_{h',u_1}}(v) &\triangleq&S_{g\oplus \mathbf{e}_{h',u_1}g}(v).
	\end{eqnarray}

	\item When $0\le u_1< h',0<u_2\le 2\ell$, pair $(0,g)\in\{0\}\times V_{u_1}$ with $(u_2,g)\in \{u_2\}\times V_{u_1}$, and reorder the elements of $S_{u_2,g}$ such that
	\begin{eqnarray}\label{Eqn_Swg_2}
		S_{u_2,g} (v)&\triangleq&S_g(v)\oplus  \mathbf{e}_{n,i_u}, \qquad g\in V_{u_1}.
	\end{eqnarray}
 According to \eqref{Eqn_S_g_2}, the symbols of $S_g$ is determined by components with indices $i_{(2\ell+1)u_1},u_1\in [0,h')$, i.e. $u_2=0$. In this case, $u_2\ne 0$, we have $\{S_g(v)\oplus  \mathbf{e}_{n,i_u}:v\in [0,2^{n-h'})\}=S_g$, $g\in V_{u_1}$, which is only a reorder of $S_{u_2,g}$ without adding or losing any elements.

	\item When $u_1= h',0\le u_2< 2\ell$, pair $(0,g)\in\{0\}\times V_{h'}$ with $(u_2+1,g)\in \{u_2+1\}\times V_{h'}$, and reorder the elements of $S_{u_2+1,g}$ such that
	\begin{eqnarray}\label{Eqn_Swg_3}
		S_{u_2+1,g} (v)&\triangleq&S_g(v)\oplus  \mathbf{e}_{n,i_u}, \qquad g\in V_{h'}.
	\end{eqnarray}	
 Similarly, \eqref{Eqn_Swg_3} is only a reorder of $S_{u_2+1,g}$ since $u_2+1\ne 0$ and $\{S_g(v)\oplus  \mathbf{e}_{n,i_u}:v\in [0,2^{n-h'})\}=S_g$, $g\in  V_{h'}$ by \eqref{Eqn_S_g_2}.
\end{itemize}

Note that $V_i,i\in [0,h']$ is the same as those sets defined in section  \ref{basic repair} such that
\begin{eqnarray}\label{Eqn_Vh_2}
V_0 \cup V_{1}\cup V_{2}\cup \cdots\cup V_{h'}=[0,2^{h'})
\end{eqnarray}
according to Lemma \ref{Lem_cover}.

\begin{Lemma}\label{Lem_Intra-groups}
	Given a failed node $i_u,u\in[0,h)$, for any $w\in [0,2\ell]$, $g\in [0,2^{h'})$, the \textbf{pairing} method above makes \eqref{Condition-0} and \eqref{Condition-1} hold.
\end{Lemma}

\begin{proof}
	During the pairing process above, suppose that group $(w,g)$ is paired with group $(w',g')$. For any $v\in [0,2^{n-h'})$, let $a=S_{w,g}(v), b=S_{w',g'}(v)$.
	
	When $0\le u_1< h',u_2=0,u=(2\ell+1)u_1$, we have $w'=w=0$, $g'=g\oplus \mathbf{e}_{h',u_1}$. According to \eqref{Eqn_S_g_2} and \eqref{order},
	\begin{eqnarray*}
		a_{i_u}=g_{u_1},b_{i_u}=g'_{u_1}=g_{u_1}\oplus1.
	\end{eqnarray*}
By \eqref{Eqn_def_s_wg}, that is
\begin{equation*}
S_{0,g'}=S_{g'}=\begin{cases}
S_g+2^{i_u}=S_{0,g}+2^{i_u}&\qquad \text{if }g_{u_1}=0,\\
S_g-2^{i_u}=S_{0,g}-2^{i_u}&\qquad \text{if }g_{u_1}=1.\\
\end{cases}
\end{equation*}
Thus, by \eqref{order}, \eqref{S_w,g}, and \eqref{Eqn_S_wg_1}, we have for any $v\in [0,2^{n-h'})$
\begin{equation*}
b=S_{0,g'}(v)=S_{g'}(v)=\begin{cases}
S_{g}(v)+2^{i_u}=S_{0,g}(v)+2^{i_u}=a+2^{i_u}, &\qquad a_{i_u}=0,\\
S_{g}(v)-2^{i_u}=S_{0,g}(v)-2^{i_u}=a-2^{i_u}, &\qquad a_{i_u}=1,\\
\end{cases}
\end{equation*}
which implies
	\begin{eqnarray*}
		b_{i_u}&=&a_{i_u}\oplus 1,\\
		b_{i}&=&a_{i}, \qquad i\in \mathcal{E}\setminus\{i_u\}.
	\end{eqnarray*}
 The preceding equations imply that \eqref{Condition-0} and \eqref{Condition-1} follow by \eqref{Eqn_lambda}.
	
Next, we prove the other two cases, i.e., $0\le u_1 < h',0<u_2\le 2\ell$ or $u_1= h',0\le u_2< 2\ell$.
 Note that for the $S_{w=0,g}$ we use the default order, i.e., $S_{w,g}(v)=S_g(v)$ for any $v\in [0,2^{n-h'}).$ Then, based on \eqref{Eqn_Swg_2} and \eqref{Eqn_Swg_3}, $b=a\oplus \mathbf{e}_{n,i_u}$ holds
 for both of the remaining two cases.
That is,
	\begin{eqnarray*}
		b_{i_u}&=&a_{i_u}\oplus 1,\\
		b_{i}&=&a_{i},\qquad i\in [0,n)\setminus\{i_u\}.
	\end{eqnarray*}
Thus, we conclude that \eqref{Condition-0} and \eqref{Condition-1} hold by \eqref{Eqn_lambda}, which completes the proof.
\end{proof}

\subsection{Repair scheme}

When $(h+1)\nmid 2^n$, the cooperative repair process of $h$ failed nodes is as follows.

\textbf{Download phase:} ~For failed node $i_u,u=(2\ell+1)u_1+u_2\in \mathcal{E}$, $u_1\in[0,h'),u_2\in[0,2\ell]$ or $u_1=h',u_2\in[0,2\ell)$, the downloaded phase is considered by three different cases. Based on Lemma  \ref{Lem_Intra-groups} and Theorem  \ref{Thm_Repair Condition},
\begin{itemize}
\item Case 1. When $0\le u_1< h',u_2=0$, node $i_u$ downloads
\begin{eqnarray}\label{DV_1}
	\{f_{j,S_{0,g}(v)}^{(0)}+f_{j,S_{0,g\oplus  \mathbf{e}_{h',u_1}}(v)}^{(0)}: g\in V_0,v\in[0,2^{n-h'}) \}
\end{eqnarray}
from helper nodes $j\in \mathcal{H}$ to recover
\begin{eqnarray*}
	\{f_{i_u,S_{0,g}(v)}^{(0)},f_{i_u,S_{0,g\oplus  \mathbf{e}_{h',u_1}}(v)}^{(0)}: g\in V_0, v\in[0,2^{n-h'})\}
\end{eqnarray*}
for itself and
\begin{eqnarray}\label{Eqn_iu_recover_25}
	\{f_{i,S_{0,g}(v)}^{(0)}+f_{i,S_{0,g\oplus  \mathbf{e}_{h',u_1}}(v)}^{(0)}: g\in V_0, v\in[0,2^{n-h'})\}
\end{eqnarray}
for other failed nodes $i\in \mathcal{E}\backslash\{i_u\}$.
\item Case 2. When $0\le u_1< h',0<u_2\le 2\ell$, node $i_u$ downloads
\begin{eqnarray}\label{DV_2}
	\{f_{j,S_{0,g}(v)}^{(0)}+f_{j,S_{u_2,g}(v)}^{(u_2)}: g\in V_{u_1}, v\in[0,2^{n-h'})\}
\end{eqnarray}
from helper nodes $j\in \mathcal{H}$ to recover
\begin{eqnarray*}
	\{f_{i_u,S_{0,g}(v)}^{(0)},f_{i_u,S_{u_2,g}(v)}^{(u_2)}: g\in V_{u_1}, v\in[0,2^{n-h'})\}
\end{eqnarray*}
for itself and
\begin{eqnarray}\label{Eqn_iu_recover_26}
	\{f_{i,S_{0,g}(v)}^{(0)}+f_{i,S_{u_2,g}(v)}^{(u_2)}: g\in V_{u_1}, v\in[0,2^{n-h'})\}
\end{eqnarray}
for other failed nodes $i\in \mathcal{E}\backslash\{i_u\}$.
\item Case 3. When $u_1= h',0\le u_2< 2\ell$, node $i_u$ downloads
\begin{eqnarray}\label{DV_3}
	\{f_{j,S_{0,g}(v)}^{(0)}+f_{j,S_{u_2+1,g}(v)}^{(u_2+1)}: g\in V_{h'}, v\in[0,2^{n-h'})\}
\end{eqnarray}
from helper nodes $j\in \mathcal{H}$ to recover
\begin{eqnarray*}
	\{f_{i_u,S_{0,g}(v)}^{(0)},f_{i_u,S_{u_2+1,g}(v)}^{(u_2+1)}: g\in V_{h'}, v\in[0,2^{n-h'})\}
\end{eqnarray*}
for itself and
\begin{eqnarray}\label{Eqn_iu_recover_u_1h'}
	\{f_{i,S_{0,g}(v)}^{(0)}+f_{i,S_{u_2+1,g}(v)}^{(u_2+1)}: g\in V_{h'}, v\in[0,2^{n-h'})\}
\end{eqnarray}
for other failed nodes  $i\in \mathcal{E}\backslash\{i_u\}$.
\end{itemize}

\vspace{3mm}
\textbf{Cooperative phase:} ~Other failed nodes $i_{\bar{u}},\bar{u}=(2\ell+1)\bar{u}_1+\bar{u}_2\in  \mathcal{E}\backslash\{i_u\}$, $\bar{u}_1\in[0,h'),\bar{u}_2\in[0,2\ell]$ or $\bar{u}_1=h',\bar{u}_2\in[0,2\ell)$ transfers
\begin{numcases}{}
	\{f_{i_u,S_{0,g}(v)}^{(0)}+f_{i_u,S_{0,g\oplus  \mathbf{e}_{h',\bar{u}_1}}(v)}^{(0)}: g\in V_0, v\in[0,2^{n-h'})\},  &if ~$0\le \bar{u}_1 < h', \bar{u}_2=0$\label{Eqn_iu_recover u_2=0}\\
	\{f_{i_u,S_{0,g}(v)}^{(0)}+f_{i_u,S_{\bar{u}_2,g}(v)}^{(\bar{u}_2)}: g\in V_{\bar{u}_1}, v\in[0,2^{n-h'})\},  &if ~$0\le \bar{u}_1 < h',0< \bar{u}_2\le 2\ell$\label{Eqn_iu_recover u_2==0}\\		
	\{f_{i_u,S_{0,g}(v)}^{(0)}+f_{i_u,S_{\bar{u}_2+1,g}(v)}^{(\bar{u}_2+1)}: g\in V_{h'}, v\in[0,2^{n-h'})\},  &if ~$\bar{u}_1 = h',0\le \bar{u}_2< 2\ell$\label{Eqn_iu_recover u_1=h'}
\end{numcases}
repaired by \eqref{Eqn_iu_recover_25}, \eqref{Eqn_iu_recover_26} or \eqref{Eqn_iu_recover_u_1h'} to node $i_u$.

\begin{itemize}
\item Case 1. When $0\le u_1 < h', u_2=0$, the failed node $i_u$ utilizes its own data $\{f_{i_u,S_{0,g}(v)}^{(0)},g\in V_0, v\in[0,2^{n-h'})\}$ to solve
\begin{eqnarray*}
	\mathcal{F}_1&=&\{f_{i_u,S_{0,g\oplus  \mathbf{e}_{h',\bar{u}_1}}(v)}^{(0)}: g\in V_0, v\in[0,2^{n-h'}), \bar{u}_1\in[0,h')\setminus\{u_1\}\},\\
	\mathcal{F}_2&=&\{f_{i_u,S_{\bar{u}_2,g}(v)}^{(\bar{u}_2)}: g\in V_{\bar{u}_1}, v\in[0,2^{n-h'}),\bar{u}_1\in[0,h'), \bar{u}_2\in [1,2\ell]\},\\
	\mathcal{F}_3&=&\{f_{i_u,S_{\bar{u}_2+1,g}(v)}^{(\bar{u}_2+1)}: g\in V_{h'}, v\in[0,2^{n-h'}),\bar{u}_2\in [0,2\ell)\}
\end{eqnarray*}
from the data in \eqref{Eqn_iu_recover u_2=0}, \eqref{Eqn_iu_recover u_2==0}, \eqref{Eqn_iu_recover u_1=h'}, respectively.

Firstly, by \eqref{Eqn_Vh_2} noting that
$$\{g, g\oplus  \mathbf{e}_{h',u_1}, g\oplus  \mathbf{e}_{h',\bar{u}_1}:g\in V_0,\bar{u}_1\in [0,h')\setminus\{u_1\}\}=\cup_{i=0}^{h'}V_i=[0,2^{h'})$$
 thus,
$$\mathcal{F}_1\cup\{f_{i_u,S_{0,g}(v)}^{(0)},f_{i_u,S_{0,g\oplus  \mathbf{e}_{h',u_1}}(v)}^{(0)}:g\in V_0,v\in[0,2^{n-h'})\}$$ contains all the symbols of $f_{i_u}^{(0)}$.

Secondly, for a fixed $0<\bar{u}_2\le 2\ell$,  we can recover $$\{f_{i_u,S_{\bar{u}_2,g}(v)}^{(\bar{u}_2)}:g\in\cup_{i=0}^{h'-1}V_i, v\in[0,2^{n-h'})\}$$
from $\mathcal{F}_2$ with $\bar{u}_1$ ranging over $[0,h')$. Combining it with $\mathcal{F}_3$, we obtain all the symbols of $f_{i_u}^{(\bar{u}_2)}$, $\bar{u}_2\in[1,2\ell]$ by \eqref{Eqn_Vh_2}.

The preceding two steps have recovered all symbols of $f_{i_u}^{(w)}$, $w\in[0,2\ell]$.

\item Cases 2 and 3. When $0\le u_1< h',0<u_2\le 2\ell$ or $u_1=h',0\le u_2< 2\ell$,  if $u_1>0$, the failed node $i_u$ utilizes its own data $\{f_{i_u,S_{0,g}(v)}^{(0)}:g\in V_{u_1},v\in[0,2^{n-h'})\}$ to solve
\begin{eqnarray*}
	\{f_{i_u,S_{0,g}(v)}^{(0)}: g\in V_0, v\in[0,2^{n-h'})\}
\end{eqnarray*}
from the data in \eqref{Eqn_iu_recover u_2=0} transferred by failed nodes $i_{\bar{u}},\bar{u}=(2\ell+1)(u_1-1)$, otherwise it does nothing. Now, the rest repair process are exactly the same as the repair of the failed node $i_{u}$, $u=(2\ell+1)u_1$ in Case1.
\end{itemize}
\vspace{3mm}	

During the download phase, by \eqref{DV_1}, \eqref{DV_2} and \eqref{DV_3}, each failed node downloads $2^{h'-m}\cdot 2^{n-h'}=2^{n-m}=N/(h+1)$ symbols from each of the $k+1$ helper nodes, i.e., the repair bandwidth at this phase is
\begin{eqnarray*}
	\gamma_1 &= h(k+1)\cdot{N\over h+1}.
\end{eqnarray*}

Next, during the cooperative phase, by \eqref{Eqn_iu_recover u_2=0}, \eqref{Eqn_iu_recover u_2==0} and \eqref{Eqn_iu_recover u_1=h'}, each failed node accesses $2^{h'-m}\cdot 2^{n-h'}=2^{n-m}=N/(h+1)$ symbols from other $h-1$ failed nodes, i.e., the repair bandwidth at this phase is
\begin{eqnarray*}
	\gamma_2 &= h(h-1)\cdot{N\over h+1}.
\end{eqnarray*}
Totally, the repair bandwidth is
\begin{eqnarray*}\label{Eqn_Repair Bandwidth2}
	\gamma &= & \gamma_1+\gamma_2 \nonumber\\
	&=&\frac{h(k+h)\times N}{h+1}
\end{eqnarray*}
which attains the optimal repair bandwidth according to \eqref{Eqn_Lower_Bound}.

This concludes the property of our repair scheme as well as the proof of the following Theorem \ref{Thm_Repair_2}.
\begin{Theorem}\label{Thm_Repair_2}
	When $(h+1)\nmid2^n$ where $h+1=(2\ell+1)2^m$, $m\ge 0$, $\ell\ge 1$, any $h\ge 2$ failed nodes of $(n,k)$ Hadamard MSR code $\mathcal{C}$ with sub-packetization $N=(2\ell+1)2^n$ can be optimally cooperative repaired with $d=k+1$ helper nodes.
\end{Theorem}

\begin{Remark}
For $(h+1)\nmid 2^n$, the  optimal cooperative repair scheme in  \cite{New-Constructions} requires  the sub-packetization $(h+1)N=(2\ell+1)2^{n+m}$. In contrast, the sub-packetization in  Theorem \ref{Thm_Repair_2} is  $2^m$ times less.
\end{Remark}

\begin{Example}Based on $(n=14,k=2)$ Hadamard  MSR code with sub-packetization $N=2^{14}$ in Example \ref{Example_1}, we obtain new $(n=14,k=2)$ Hadamard MSR code $\mathcal{C}$ with sub-packetization $N=3\times2^{14}$ by generating 3 instances. At this point, $m=2, \ell=1,h'=3$, the number of failed nodes is $h=11$. Assume that the failed nodes are the first $11$ nodes, i.e., $\{i_0,i_1,\cdots,i_h\}=\{0,1,\cdots,10\}$, and the remaining $3$ nodes are helper nodes.
		
		Firstly, divide the $2^{14}$ symbols of the base code into $8$ groups based on \eqref{Eqn_S_g_2}, each group having $2^{11}$ symbols. Then 3 instances are divided into $M=3\times 8$ groups. For any failed node $i_u\in\mathcal{E}$, the pairing method is shown at the table  \ref{non-Divisible-group}, where $\{i\}\times V_j$ is group $(i,g),g\in V_j$.
	\begin{table}[h]
		\begin{center}
			\caption{the pairing method for $h=11$ failed nodes}
			\label{non-Divisible-group}
			\setlength{\tabcolsep}{6pt}
			\begin{tabular}{|c|c|c|c|c|c|c|}
				\hline
				node $u$& 0 & 1 & 2 & 3 & 4 & 5 \\
				\hline
				Pairing &  $\{0\}\times V_0,\{0\}\times V_1$	&  $\{0\}\times V_0,\{1\}\times V_0$ &  $\{0\}\times V_0,\{2\}\times V_0$ & $\{0\}\times V_0,\{0\}\times V_2$ &  $\{0\}\times V_1,\{1\}\times V_1$ & $\{0\}\times V_1,\{2\}\times V_1$  \\
				\hline
				node $u$ & 6 & 7 & 8  & 9 & 10 & \\
				\hline
				Pairing &  $\{0\}\times V_0,\{0\}\times V_3$ &  $\{0\}\times V_2,\{1\}\times V_2$ &  $\{0\}\times V_2,\{2\}\times V_2$ & $\{0\}\times V_3,\{1\}\times V_3$ & $\{0\}\times V_3,\{2\}\times V_3$&\\
				\hline
			\end{tabular}
		\end{center}
	\end{table}

	In the above table, to repair nodes $1,2,4,5,7$ and $8$ reorder the symbols of $(w,g)$ by \eqref{Eqn_Swg_2} and node $9$ and $10$ by \eqref{Eqn_Swg_3}. The reorder processes for these nodes are shown in Table \ref{reorder}.
		\begin{table}[h]
		\begin{center}
			\caption{the reorder of the symbols of node $i\in [0,10]\backslash \{0,3,6\},$ where $S_g\oplus\mathbf{e}_{n,u}=\{S_g(v)\oplus \mathbf{e}_{n,u}:v\in [0,2^{11})\}$}
			\label{reorder}
			\setlength{\tabcolsep}{7pt}
			\begin{tabular}{|c|c|c|c|c|c|c|}
				\hline
				node $u$& 1 & 2 & 4 & 5  & 7 & 8 \\
				\hline
				\multirow{2}{*}{Reorder} &$ S_{1,0}=S_0\oplus\mathbf{e}_{14,1}$ &$ S_{2,0}=S_0\oplus\mathbf{e}_{14,2}$ & $S_{1,1}=S_1\oplus\mathbf{e}_{14,4}$ & $S_{2,1}=S_1\oplus\mathbf{e}_{14,5}$ & $S_{1,2}=S_2\oplus\mathbf{e}_{14,7}$ & $S_{2,2}=S_2\oplus\mathbf{e}_{14,8}$ \\
				&$ S_{1,7}=S_7\oplus\mathbf{e}_{14,1}$ &$ S_{2,7}=S_7\oplus\mathbf{e}_{14,2}$ &$ S_{1,6}=S_6\oplus\mathbf{e}_{14,4}$ & $S_{2,6}=S_6+\mathbf{e}_{14,5}$ & $S_{1,5}=S_5\oplus\mathbf{e}_{14,7}$ & $S_{2,5}=S_5\oplus\mathbf{e}_{14,8}$  \\
				\hline
				node $u$ & 9 & 10 &&&&\\
				\hline
				\multirow{2}{*}{Reorder} & $S_{1,4}=S_4\oplus\mathbf{e}_{14,9}$ & $S_{2,4}=S_4\oplus\mathbf{e}_{14,10}$ &&&&\\
				& $S_{1,3}=S_3\oplus\mathbf{e}_{14,9}$ & $S_{2,3}=S_3\oplus\mathbf{e}_{14,10}$&&&&\\
				\hline
			\end{tabular}
		\end{center}
	\end{table}

	When $u_2=0$, the data of node $0,3,6$ can be paired and recovered following the previous section. We will only illustrate the repair process of the failed nodes when $u_2>0$.
	
	Take symbols $S_{0,0}(1)=0$ and $S_{1,0}(1)=2$ of node 1 as an example, applying  conditions \eqref{Condition-0}, \eqref{Condition-1} and equation \eqref{Eqn_C_21}, we have
	\begin{eqnarray}\label{E5}
		\lambda_{1,0}^tf_{1,0}^{(0)}+\lambda_{1,2}^tf_{1,2}^{(1)}+\sum\limits_{i\ne1}(\lambda_{i,0}^tf_{i,0}^{(0)}+\lambda_{i,2}^tf_{i,2}^{(1)})=0, \qquad t\in[0,11].
	\end{eqnarray}
	 From \eqref{Eqn_lambda},
	$$\lambda_{1,0}\ne \lambda_{1,2},\lambda_{i,0}=\lambda_{i,2},i\ne 1.$$
	
	There are 12 unknowns $f_{1,0}^{(0)},f_{1,2}^{(1)},f_{i,0}^{(0)}+f_{i,2}^{(1)},i\in\{0\}\cup [2,10]$ in equation \eqref{E5}. With $t$ enumerating $t\in[0,11]$, the equation has unique solution, so the desired data could be solved.
	
	To avoid duplication, readers are encouraged to validate the download and cooperative phases by following the analysis we presented.
\end{Example}

\section{Conclusion}\label{Conclusion}

In this paper, a new cooperative repair scheme  was proposed for $(n,k)$ Hadamard MSR code with $h$ node failures, which downloads the data from $d=k+1$ helper nodes. Particularly, the cooperative repair scheme is feasible for sub-packetization   $2^n$ when $(h+1)|2^n$ and $(2\ell+1)2^n$ when $h+1=(2\ell+1)2^m$ for $m\ge 0$ and $\ell\ge 1$. In contrast to the known best result in \cite{New-Constructions}, the sub-packetization is greatly decreased.

Unfortunately, our scheme  only works for $d=k+1$ helper nodes to recover the failed data, and can not be generalized to arbitrary $k\le d\le n-h$. Therefore, our next task is to find  exact cooperative repair scheme for any $d$ helper nodes with small sub-packetization.

\end{document}